\documentclass[11pt]{article} %
\usepackage[margin=1in]{geometry} %

\usepackage[backend=bibtex,citestyle=verbose,style=numeric-comp,natbib=true,maxnames=10,maxalphanames=20,sortcites=true,doi=false,isbn=false]{biblatex} %

\usepackage{amsmath}            %
\usepackage{amsthm}             %

\usepackage{microtype}          %
\usepackage[T1]{fontenc}
\usepackage[full]{textcomp}
\usepackage{newtxtext}
\usepackage[vvarbb]{newtxmath} %
\usepackage[scaled=0.95]{inconsolata} %

\usepackage{article} %
\usepackage{math} %
\usepackage{algo} %
\usepackage{cuts} %

\usepackage{calc}%
\usepackage{wrapfig} %
\usepackage{placeins} %
\usepackage{graphicx} %
\usepackage{layout} %
\usepackage{setspace} %
\usepackage{mathtools} %
\usepackage{grffile} %
\usepackage{pdfpages} %
\usepackage{mdframed} %
\usepackage{bm} %
\usepackage{needspace} %
\usepackage[makeroom]{cancel}   %
\usepackage{pbox}               %
\usepackage{varwidth}           %
\usepackage{graphbox}           %
\usepackage{titlesec}

\newcommand{\chapter}[1]{\title{#1} \maketitle}

\newcommand{\defterm}[1]{\emph{#1}}

\makeatletter
\g@addto@macro\bfseries{\boldmath}
\makeatother

\titlespacing*{\paragraph}{%
  0pt}{
  {\medskipamount}}{
  1em}

\bibliography{references}

\begin{document}

\author{Kent Quanrud\thanks{Dept.\ of Computer Science, Purdue University,
    West Lafayette, IN 47907. {\tt krq@purdue.edu}.}}

\title{Fast~Approximations for~Rooted~Connectivity
  in~Weighted~Directed~Graphs}

\maketitle

\begin{abstract}
  We consider approximations for computing minimum weighted cuts in
  directed graphs.  We consider both rooted and global minimum cuts,
  and both edge-cuts and vertex-cuts. For these problems we give
  randomized Monte Carlo algorithms that compute a
  $(1+\epsilon)$-approximate minimum cut in $\tilde{O}(n^2 / \epsilon^2)$
  time. These results extend and build on recent work
  \cite{cq-simple-connectivity} that obtained exact algorithms with
  similar running times in directed graphs with small integer
  capacities.
\end{abstract}

\providecommand{\ectime}{\fpair{\operatorname{EC}}} %
\providecommand{\vctime}{\fpair{\operatorname{VC}}} %
\providecommand{\ec}{\lambda} %
\providecommand{\vc}{\kappa} %
\providecommand{\inneighbors}{\fparnew{N^{-}}} %
\providecommand{\polyeps}{\eps^{\bigO{1}}} %
\providecommand{\inneighbors}{\fparnew{N^-}} %

\providecommand{\sparseG}{G_0}
\providecommand{\sparseV}{V_0}
\providecommand{\sparseE}{E_0}

\providecommand{\optS}{S^{\star}}     %
\providecommand{\weight}{\fparnew{w}} %
\providecommand{\apxvc}{\vc'}         %
\providecommand{\apxsinksize}{k'}     %

\providecommand{\optec}{\lambda^{\star}} %
\providecommand{\revG}{G_{\text{rev}}} %
\providecommand{\splitG}{G_{\text{split}}}%
\providecommand{\tout}{t^+}    %
\providecommand{\rin}{r^-}          %
\providecommand{\vin}{v^-}
\providecommand{\vout}{v^+}         %
\providecommand{\auxvin}{a_v^-}         %
\providecommand{\auxvout}{a_v^+}         %
\providecommand{\uout}{u^+}
\providecommand{\optSink}{T^{\star}} %
\providecommand{\Sink}{T} %
\providecommand{\optsink}{\optSink} %
\renewcommand{\defterm}[1]{\emph{#1}}

\clearpage

\section{Introduction}

Let $G = (V,E)$ be a directed graph with $m$ edges and $n$ vertices.
Let $G$ have positive edge weights $w: E \to \preals$.  Recall that
$G$ is \emph{strongly connected} if any vertex can reach any other in
the graph.  The (global, weighted) edge connectivity is the minimum
weight of edges that needs to be removed so that $G$ is no longer
strongly connected.  The minimum (weight) cut is the corresponding set
of edges.  Determining the edge connectivity and computing the minimum
cut are basic problems in graph algorithms.  This work develops a
faster randomized algorithm for approximating the minimum weight cut.

The algorithm for edge connectivity is developed alongside for the
following related connectivity problems also of basic interest.  Let
$r \in V$ be a fixed vertex, called the \emph{root}.  The
\emph{minimum rooted cut} (from $r$), also called the \emph{minimum
  $r$-cut}, is the minimum weight set of edges whose removal
disconnects $r$ from at least one vertex.  Global connectivity follows
from rooted connectivity by choosing any root arbitrarily, and
computing the rooted connectivity in both $G$ and the reversed graph.
Most of the algorithmic discussion in this work is focused on rooted
connectivity and global connectivity is obtained as a by-product.
Rooted connectivity has other connections in combinatorial
optimization \cite{schrijver-book,frank}; for example,
\citet{edmonds-70} showed that the $r$-rooted edge connectivity equals
the maximum number of arboresences rooted at $r$ that can be packed
into the graph.  Now suppose instead that the graph has vertex weights
$w: V \to \preals$. The \emph{(global) vertex connectivity} is the
minimum weight of vertices that needs to be removed so that $G$ either
is no longer strongly connected or consists of only a single vertex;
the minimum vertex cut is the corresponding set of vertices.  One can
also define rooted vertex connectivity analogously to rooted edge
connectivity.  This work also develops fast approximation algorithms
for rooted and global vertex connectivity.

These connectivity problems are well-studied and we first give an
overview of classical results, with particular focus on algorithms for
directed graphs, before discussing more recent developments.  There is
a long line of algorithms for directed edge connectivity
\cite{schnorr-79,timofeev-82,matula-87,mansour-schieber,gabow-95,hao-orlin}
(see also \cite{schrijver-book}), of which we highlight the most
pertinent.  For general weights, an algorithm by \citet{hao-orlin}
finds the minimum rooted cut in $\bigO{m n \log{n^2 / m}}$ time.  For
multigraphs, \citet{gabow-95} gives an
$\bigO{m \lambda \log{n^2 / m}}$ time algorithm for the minimum
$r$-cut, where $\lambda$ is the weight of the minimum rooted cut.
Directed vertex connectivity likewise has had many algorithms, and
many of these running times are parametrized by the weight of the
vertex cut
\cite{podderyugin-73,even-tarjan,cheriyan-reif,galil-80,hrg-00,gabow-06}.
Of those that are not, we highlight the randomized
$\bigO{m n \log{n}}$ time algorithm of \citet{hrg-00} that remains the
fastest algorithm in weighted and directed graphs.

Recently there has been a flurry of results for graph algorithms
several of which impact directed connectivity.  There have been many
significant developments for $(s,t)$-flow for both edge- and
vertex-capacitated directed graphs.
\cite{goldberg-rao,orlin-13,lee-sidford,madry-13,madry-16,liu-sidford-20,kls-20,brand+20,
  brand+21,glp-21,chuzhoy-khanna-19,nsy-19}.  Very recently
\cite{brand+21} obtained an $\apxO{m + n^{1.5}}$ running time for
edge-capacitated $(s,t)$-max flow, generalizing a preceding
$\apxO{m + n^{1.5}}$ algorithm for vertex capacitated $(s,t)$-max flow
\cite{brand+20}.\footnote{Here and throughout $\apxO{\cdots}$ hides
  polylogarithmic factors. The algorithms in this work use the
  $\apxO{m + n^{1.5}}$ time algorithms \cite{brand+20,brand+21} as a
  subroutine which incur large polylogarithmic factors hidden in the
  $\apxO{\cdots}$ notation. Consequently we generally do not try to
  optimize polylogarithmic factors in this article.}
Another recent development is a randomized $\apxO{m \vc^2}$ time exact
algorithm and a randomized $\apxO{m \vc / \eps}$ time
$\epsmore$-approximation for global vertex connectivity in unweighted
directed graphs \cite{forster+20}.  These algorithms are based on
local algorithms for vertex connectivity and influence the local
algorithms that appear in this work.  A very recent and independent
work of \cite{li+21} has obtained an $\apxO{m n^{1-1/12 + o(1)}}$ time
algorithm for vertex connectivity in directed and unweighted
graphs. (We have not yet had time to digest and make a proper
comparison to \cite{li+21}.)  The last recent work we discuss is a
randomized, $\apxO{n^2 U^2}$ time (exact) algorithm for rooted and
global edge connectivity in directed graphs with small integer
capacities between $1$ and $U$
\cite{cq-simple-connectivity}. \cite{cq-simple-connectivity} also
gives a $\apxO{\vc n W}$-time exact algorithm for rooted and global
vertex capacity with integer weights, where $W$ is the total weight in
the graph, and $\vc$ is the weight of the minimum vertex cut.
\cite{cq-simple-connectivity} introduces elementary ideas to sparsify
rooted connectivity problems but the crux of the argument needs the
capacities to be small.  The driving motivation of this article is to
overcome the limitations of \cite{cq-simple-connectivity} to small
integer capacities and extend the ideas to the weighted setting.

\subsection{Results.} The primary results of this work extend the
$\apxO{n^2}$ randomized running times of
\cite{cq-simple-connectivity} to the weighted setting while
allowing for approximation. These algorithms take as input an
additional parameter $\eps > 0$; the goal is to compute a cut
whose weight is at most a $\epsmore$-multiplicative factor
greater than the minimum cut.

\paragraph{Edge connectivity.} The first result is for rooted and
global edge connectivity.
\begin{restatable}{theorem}{ApxEC}
  \labeltheorem{apx-ec} Given a directed graph with polynomially
  bounded weights, and $\eps \in (0,1)$, a $\epsmore$-approximate
  minimum rooted or global edge cut can be computed with high
  probability in $\apxO{n^2 / \eps^2}$ randomized time.
\end{restatable}
It is of theoretical interest to obtain $o(mn)$ running times due to
the longstanding $\bigO{m n \log{n^2 / m}}$ running time of
\cite{hao-orlin} and the connection to flow decompositions.  Observe
that the above running time for edge connectivity is $o(mn)$ except
for the case where $m \leq n^{1+ o(1)}$.  One can modify the algorithm
in \reftheorem{apx-ec} (leveraging, in particular, the recent
$m^{1.5 - 1/28}$ time algorithm for $(s,t)$-flow \cite{glp-21}) to
establish a $o(mn)$ running time, as follows.
\begin{restatable}{corollary}{FlowBarrierApxEC}
  \labelcorollary{flow-barrier-apx-ec} There exists a constant
  $c > 0$ such that, for all fixed $\eps > 0$, an
  $\epsmore$-approximate minimum weight rooted or global cut in
  a directed graph with polynomially bounded edge weights can be
  computed with high probability in
  $\bigO{(mn)^{1-c} / \eps^{2}}$ randomized time.
\end{restatable}

\paragraph{Vertex connectivity.}
The other main result is for rooted vertex connectivity. Here,
$\outdegree{v}$ denotes the unweighted out-degree of a vertex $v$.
\begin{restatable}{theorem}{ApxRVC}
  \labeltheorem{apx-rvc} Let $\eps \in (0,1)$, let $G = (V,E)$ be a
  directed graph with polynomially bounded vertex weights, and let
  $r \in V$ be a fixed root. A $\epsmore$-approximate minimum vertex
  $r$-cut can be computed with high probability in
  $\apxO{m + n (n - \outdegree{r}) / \eps^2}$ randomized time.
\end{restatable}
An argument by \cite{hrg-00} (with some modifications) implies that
the above running time for rooted vertex connectivity, combined with
randomly sampling roots by weight, leads to the following running time
for global vertex connectivity.
\begin{restatable}{corollary}{ApxVC}
  \labelcorollary{apx-vc} For all $\eps \in (0,1)$, a
  $\epsmore$-approximate minimum weight global vertex cut in a
  directed graph with polynomially bounded vertex weights can be
  computed with high probability in $\apxO{n^2 / \eps^2}$ expected
  time.
\end{restatable}
In the same way as for edge connectivity above, \reftheorem{apx-rvc}
and \refcorollary{apx-vc} leads to $o(mn)$ time approximation
algorithms for rooted and global vertex connectivity, as follows.
\begin{restatable}{corollary}{FlowBarrierApxVC}
  \labelcorollary{flow-barrier-apx-vc} There exists a constant
  $c > 0$ such that, for all fixed $\eps > 0$, an
  $\epsmore$-approximate minimum weight rooted or global cut in
  a directed graph with polynomially bounded edge weights can be
  computed with high probability in
  $\bigO{(mn)^{1-c} / \eps^{2}}$ randomized time.
\end{restatable}


\subsection{Key ideas.}

The high level approach is inspired by previous work in
\cite{cq-simple-connectivity}, which was limited to small integer
capacities.  Let us focus on edge connectivity as the ideas for vertex
connectivity are similar.  An important idea that emerges from
\cite{cq-simple-connectivity} is that there are useful tradeoffs based
on the number of vertices in the sink component.  Let $k$ denote the
number of vertices in the sink component of the minimum $r$-cut, and
for simplicity, suppose $k$ is known.  \cite{cq-simple-connectivity}
observed that if $k$ is small, and the graph has small capacities,
then the graph can be sparsified by contracting vertices with
in-degree greater than $\bigO{k}$ into the root.  Meanwhile, if $k$ is
large, then it is easier to sample a vertex from the sink component
and then apply $(s,t)$-flow.  Balancing these tradeoffs leads to an
$\apxO{n^2}$ time exact algorithm for rooted edge connectivity (for
small integer capacities).  Note that the sampling approach for large
$k$ extends to the weighted setting. However the sparsification
argument requires the assumption of small capacities.  The high-level
goal of this work is to extend the ideas from
\cite{cq-simple-connectivity} to the weighted setting.

The first step is to take advantage of the approximation error and
discretize the edge weights (and also reduce the number of edges) by
random sampling.  While random sampling is known to preserve cuts in
undirected graphs \cite{karger, benczur-karger, spielman-srivastava},
there are no such guarantees in directed graphs.  Here we continue the
theme of balancing tradeoffs in the size of the sink component, $k$.
Rather than sampling as to preserve all cuts, we only sample to
preserve the in-cuts of vertex sets of size less than or equal to $k$
of the minimum rooted cut (which can be guessed).  This requires
paying some overhead in proportion to the target component size.  To
address large vertex sets for which the random sampling might
drastically understimate the in-cut, we add appropriately weighted
auxiliary edges from the root to every vertex. The auxiliary edges
make it impossible for large vertex sets to induce the minimum weight
rooted cuts, but also have limited impact on small vertex sets. (This
approach is inspired in part by the sparsification ideas in
\cite{cq-simple-connectivity} and in another part by the pessimistic
estimator in \cite{li-21}.)  Thereafter one can contract high-degree
nodes in the root similar to \cite{cq-simple-connectivity}.  The end
result is a sparser graph where the sparsity depends on $k$.
Moreover, up to scaling, the sparsified graph has integer capacities,
and the weight of the minimum $r$-cut becomes proportional to $k$.

In this sparsified graph, depending on the (guessed) size of the sink
component, the algorithm pursues one of two options.  If $k$ is
smaller than (roughly) $\sqrt{n}$, the graph is very sparse and the
size of the cut is small, and we run an algorithm based on a new
deterministic local cut algorithm that takes advantage of the small
value $k$. If $k$ is larger, then we try to sample a vertex from the
sink component and run $(s,t)$-flow.  In both cases we work in the
sparsified graph.  Balancing terms between the running times of these
two approaches leads to the claimed $\apxO{n^2 / \eps^2}$ running
time.

We highlight that in previous work for small capacities, in the regime
where $k$ is small, \cite{cq-simple-connectivity} was able to use
Gabow's algorithm (on a sparsified graph) to find the minimum rooted
cut. Here, in the presence of capacities (even after sparsification),
Gabow's algorithm has a larger polynomial dependency on $k$ than
desired and one needs new ideas.  The local cut approach developed
here is inspired by the recent randomized algorithms of
\cite{forster+20,cq-simple-connectivity}. Compared to these previous
works, the new local cut algorithm has a better dependency on $k$ and
$\eps$ and is also deterministic.  The local cut algorithm routes flow
in the reversed graph from a fixed vertex $t$ to the root, with
localized running times depending only on $k$.  It takes advantage of
auxiliary edges from the root added in the sparsification step to find
short augmenting paths.  The improved running time comes from a
refined analysis based on how many of these auxiliary edges have been
saturated.  The fact that we are always routing flow to the root also
removes the guess work from \cite{forster+20,cq-simple-connectivity}
and makes the algorithm deterministic.\footnote{That said, the overall
  algorithm for rooted edge connectivity is still randomized.}

\paragraph{Organization.}
The remainder of this article is divided into two
sections. \refsection{apx-ec} considers edge connectivity, and proves
\reftheorem{apx-ec} and
\refcorollary{flow-barrier-apx-ec}. \refsection{apx-vc} considers
vertex connectivity, and proves \reftheorem{apx-rvc},
\refcorollary{apx-vc}, and \refcorollary{flow-barrier-apx-vc}.


\section{Rooted edge connectivity}

\labelsection{apx-ec}

In this section, we design and analyze an $\apxO{n^2 / \eps^2}$-time
approximation algorithm for the minimum weight rooted edge cut. We
present the algorithm as three main steps. Each of the steps are
parameterized by values $\lambda > 0$ and $k \in \naturalnumbers$
that, in principle, are meant to be constant factor estimates for the
weight of the minimum rooted cut and the number of vertices in the
sink component of the minimum rooted cut, respectively. The first step
is a sparsification result that (assuming $\lambda$ and $k$ are
accurate) reduces the problem to a rooted graph with roughly $nk$
edges and rooted connectivity roughly $k$ in addition to a few other
helpful properties. This sparsification procedure is used by both of
the remaining two steps.  The second step, preferable for small $k$,
approximates the minimum rooted cut in roughly $nk^2$ time, and is
based on a new deterministic local cut algorithm that makes essential
use of some of the specific properties of the sparsification
lemma. The third step, preferable for large $k$, approximates the
minimum rooted cut in roughly $n^2 + n^{2.5}/k$ time, via random
sampling and $(s,t)$-flow in the sparsified graph. Balancing terms
leads to the claimed running time.

\paragraph{Sparsification.}
Our first lemma sparsifies the graph while preserving the minimum
rooted cut. The algorithm is parameterized by a target number of
vertices in the sink component and the sparsity of the output graph
depends on this input parameter.
\begin{lemma}
  \labellemma{apx-rec-sparsification} Let $G = (V,E)$ be a directed
  graph with positive edge weights.  Let $r \in V$ be a fixed root
  vertex.  Let $\eps \in (0,1)$, $\ec > 0$, and
  $k \in \naturalnumbers$ be given parameters.  In randomized linear
  time, one can compute a randomized directed and edge-weighted graph
  $\sparseG = (\sparseV, \sparseE)$, where $\sparseV \subseteq V$ and
  $r \in \sparseV$, and a scaling factor $\tau > 0$, with the
  following properties.
  \begin{properties}
  \item \label{recs-weights} $\sparseG$ has integer edge weights
    between $1$ and $\bigO{k \log{n} / \eps^2}$.
  \item \label{recs-degree} Every vertex $v \in \sparseV$ has unweighted in-degree at most
    $\bigO{k \log{n} / \eps^2}$ in $\sparseG$.
  \item \label{recs-aux-edge} For every $v \in \sparseV - r$ there is
    an edge $(r,v)$ with capacity at least
    $\bigOmega{\log{n} / \eps}$.
  \item \label{recs-cut-weights} With high probability, for all
    $S \subseteq \sparseV - r$, the weight of the in-cut induced by
    $S$ in $\sparseG$ (up to scaling by $\tau$) is at least the
    minimum of $\epsless$ times the weight of the induced in-cut in
    $G$ and $c \lambda$ for any desired constant $c > 1$, and at most
    $\epsmore$ times of the weight of the induced in-cut $G$ plus
    $\eps \lambda \sizeof{S} / k$.
  \item \label{recs-cut-sets} With high probability, for all
    $S \subseteq V - r$ such that $\sizeof{S} \leq k$ and the weight
    of the induced in-cut is $\leq \bigO{\lambda}$, we have
    $S \subseteq \sparseV$.
  \end{properties}
  In particular, if the minimum $r$-cut has weight
  $\bigTheta{\lambda}$, and the sink component of a minimum $r$-cut
  has at most $k$ vertices, then with high probability $\sparseG$
  preserves the minimum $r$-cut up to a $\apxmore$-multiplicative
  factor.
\end{lemma}

\begin{proof}
  Consider the following randomized algorithm applied to the input
  graph $G$.
  \begin{quote}
    \begin{steps}
    \item \labelstep{rec-first-step} Let
      $\tau = c_{\tau} \eps^2 \lambda / k \log{n}$ and
      \begin{math}
        \Delta = 
        c_\Delta k \log{n} / \eps^2
      \end{math}
      for a sufficiently small constant $c_\tau > 0$ and a sufficiently
      large constant $c_\Delta > 0$.
    \item \labelstep{rec-importance-sample} Importance sample each
      edge weight to be a discrete multiple of $\tau$. Drop any edge
      with weight $0$.
    \item \labelstep{rec-extra-edges} Add an edge of weight
      $\eps \lambda / 2 k$ from the root to every vertex.
      \commentitem{Decreasing $c_{\tau}$ and $\eps$ as needed, we
        assume $\ec$ and $\eps \ec / 2k$ are multiples of $\tau$}.
    \item \labelstep{rec-rescale} Scale down all edge weights by
      $\tau$ (which makes them integers).
    \item \labelstep{rec-truncate} Truncate all edge weights to be at
      most $c_w k \log{n} / \eps^2$ for a sufficiently large constant
      $c_w > 0$ (while maintaining integrality).
    \item \labelstep{rec-contract} For any vertex $v$ with unweighted
      in-degree $\geq \Delta$, contract $v$ into $r$.
    \end{steps}
  \end{quote}
  Consider the graph $\sparseG$ obtained by the above steps.  Of the
  claimed properties, \ref{recs-weights}, \ref{recs-degree}, and
  \ref{recs-aux-edge} follow directly from the construction. The
  remaining proof is dedicated to proving the high-probability claims
  in \ref{recs-cut-weights} and \ref{recs-cut-sets}.  We first show
  that the initial steps \refstep{rec-first-step} to
  \refstep{rec-extra-edges} -- before rescaling -- preserves the
  weights of the $r$-cuts in the sense of \ref{recs-cut-weights}
  (without the rescaling). We then analyze the remaining steps which
  rescale and contract the graph.

  For each set $S$, let $f(S)$ denote the weight of the in-cut at $S$
  in $G$. Let $g(S)$ denote the randomized weight of the in-cut after
  step \refstep{rec-importance-sample}.  Let $h(S)$ denote the
  randomized weight of the in-cut of $S$ after adding the auxiliary
  edges in \refstep{rec-extra-edges}.  The first claim analyzes the
  concentration of $g(S)$ for all sets $S$.

  \begin{claims}
  \item \labelclaim{rec-concentration} With high probability, for all
    $S \subseteq V$
    \begin{align*}
      \absvof{g(S) - f(S)} \leq \eps f(S) + \frac{\eps \lambda
      \sizeof{S}}{2 k}.
    \end{align*}
  \end{claims}
  The above claim consists of an upper and lower bound on $g(S)$ for
  all $S$.  We first show the lower bound on $g(S)$ holds for all $S$
  with high probability.  Fix $S \subseteq V$.  $g(S)$ is an
  independent sum with expected value $f(S)$ and where each term in
  the sum is nonnegative and varies by at most $\tau$. By a variation
  of standard Chernoff
  inequalities\footnote{\label{footnote:additive-chernoff} Here we
    apply the following bounds (appropriately rescaled) which follow
    from the same proof as the standard multiplicative Chernoff bound.
    \begin{quote}
      \itshape Let $X_1,\dots,X_n \in [0,1]$ independent random
      variables. Then for all $\eps > 0$ sufficiently small and all
      $\gamma > 0$,
      \begin{align*}
        \probof{X_1 + \cdots + X_n \leq \epsless \evof{X_1 + \cdots + X_n} -
        \gamma} \leq e^{-\eps \gamma}
      \end{align*}
      and
      \begin{align*}
        \probof{X_1 + \cdots + X_n \geq \epsmore \evof{X_1 + \cdots + X_n} +
        \gamma} \leq e^{-\eps \gamma}.
      \end{align*}
    \end{quote}
  }, for any $\gamma \geq 0$, we have
  \begin{align*}
    \probof{g(S) \leq \epsless  f(S) - \gamma} \leq e^{- \eps \gamma /
    \tau}
    =                           %
    n^{- \gamma k \log{n} / c_{\tau} \eps \lambda},
  \end{align*}
  In particular, for $\gamma = \eps \lambda \sizeof{S} / 2 k$, the RHS
  is at most
  \begin{math}
    n^{- c_0 \sizeof{S}}
  \end{math}
  where $c_0$ is a constant under our control (via $c_{\tau}$). For
  large enough $c_0$, we can take the union bound over all sets of
  vertices.  This establishes that the lower bounds for $g(S)$ hold
  for all $S$ with high probability. The upper bounds also hold with
  high probability by a symmetric argument.

  Now we analyze the in-cuts after step
  \refstep{rec-extra-edges}.
  Recall that for $S \subseteq V$, $h(S)$ denotes the weight of
  the in-cut of $S$ after adding the auxiliary edges in
  \refstep{rec-extra-edges}.
  \begin{claims}
  \item \labelclaim{rec-padded-concentration} With high
    probability, for all $S \subseteq V - r$, we have
    \begin{align*}
      \epsless f(S) \leq h(S) \leq \epsmore f(S) + \eps \lambda
      \sizeof{S} / k.
    \end{align*}
  \end{claims}
  Indeed, we have $h(S) = g(S) + \eps \lambda \sizeof{S} / 2k$ for all
  $S \subseteq V - r$. The additive term introduced by $h$ offsets the
  additive error in the lower bound on $g(S)$ in
  \refclaim{rec-concentration}. This term also adds on to the additive
  error in the upper bound of \refclaim{rec-concentration} for a total
  of $\eps \lambda \sizeof{S} / k$. Thus in the high probability event
  of \refclaim{rec-concentration}, we have the bounds described by
  \refclaim{rec-padded-concentration} for all $S$.

  Henceforth, let us assume that the high probability event in
  \refclaim{rec-padded-concentration} holds. (Otherwise the algorithm
  fails.)  \refclaim{rec-padded-concentration} implies that, after
  step \refstep{rec-extra-edges}, the cuts in $\sparseG$ preserve the
  weight of the cuts in $G$ in the approximate sense of
  \ref{recs-cut-weights} (without the scaling).  Now, after step
  \refstep{rec-extra-edges}, all the weights are divisible by $\tau$.
  After scaling down by $\tau$ in step \refstep{rec-rescale}, we will
  continue to preserve the $r$-cuts in the desired sense (up to
  scaling). Truncating weights in $\sparseG$ to
  $\bigO{\lambda / \tau}$ in \refstep{rec-truncate} decreases the
  weight of some cuts, but to no less than $\bigO{\lambda /
    \tau}$. The contractions in step \refstep{rec-contract} only
  removes some $r$-cuts from consideration and does not effect the
  weight of any remaining cuts This establishes property
  \ref{recs-cut-weights}.

  To show that the contractions in step \refstep{rec-contract}
  preserve property \ref{recs-cut-sets}, let $\Sink$ be the sink
  component of any $r$-cut of capacity $\leq \bigO{\lambda}$ and with
  $\sizeof{\Sink} \leq k$. We claim that any vertex in $\Sink$, has
  unweighted in-degree at most $\Delta$ in the graph obtained after
  \refstep{rec-rescale}. Indeed, fix any such vertex $v \in \Sink$,
  and consider the edges going into $v$. At most $k - 1$ of those
  edges can come from another vertex in $\Sink$, since $\Sink$ has at
  most $k$ vertices. The remaining edges must be in the in-cut of
  $\Sink$, and the in-cut of $\Sink$ has at most
  $\bigO{\lambda / \tau} = \bigO{k \log{n} / \eps^2}$ edges (per
  property \ref{recs-cut-weights}). Thus there are less than
  $\Delta = \bigO{k \log{n} / \eps^2}$ edges incident to $v$. In
  conclusion, any vertex $v$ with in-degree more than $\Delta$ lies
  outside $\Sink$ and can be safely contracted into the root. This
  establishes property \ref{recs-cut-sets} and completes the proof.
\end{proof}

\paragraph{Rooted edge connectivity for small sink components.}

This section presents an approximation algorithm for rooted vertex
connectivity for the particular setting where the sink component is
small. In particular, we are given an upper bound $k$ on the number of
vertices in the sink component, and want to obtain running times of
the form $n \poly{k}$.  When a similar situation arose previously for
small integer capacities in \cite{cq-simple-connectivity},
\cite{cq-simple-connectivity} used Gabow's algorithm which works well
for unweighted multigraphs. Here, while
\reflemma{apx-rec-sparsification} produces relatively sparse graphs
with integral edge capacities, the edge capacities imply a multigraph
with roughly $n k^2$ edges, and Gabow's algorithm would then take
roughly $nk^3$ time.  This section develops an alternative approach
that reduces the dependency on $k$ to $k^2$, and is inspired by
existing local algorithms for (global and rooted) vertex cuts
\cite{forster+20,cq-simple-connectivity}.  Compared to
\cite{forster+20,cq-simple-connectivity}, the algorithm here is for
edge cuts and is designed to take full advantage of the properties of
the graph produced by \reflemma{apx-rec-sparsification}.  These
modifications have some tangible benefits.  First, it improves the
dependency on $k$ and $\epsilon$. (We estimate that previous
approaches lead to an $\apxO{n k^{3}/\eps^{5}}$ running time.)
Second, the local subroutine here is deterministic whereas before they
were randomized.  Third and last, as suggested by the better running
time and the determinism, the version presented here is arguably
simpler and more direct then the previous algorithms (for this
setting).

\begin{lemma}
  \labellemma{local-apx-rec} Let $G = (V,E)$ be a directed graph with
  positive edge weights $w: V \to \preals$. Let $r \in V$ be a fixed
  root vertex.  Let $\eps \in (0,1)$, $\ec > 0$ and
  $k \in \naturalnumbers$ be given parameters.  There is a randomized
  linear time Monte Carlo algorithm that, with high probability,
  produces a deterministic data structure that supports the following
  query.

  For $t \in V$, let $\ec_{t,k}$ denote the weight of the minimum
  $(r,t)$-cut such that the sink component has at most $k$ vertices.
  Given $t \in V$, deterministically in $\bigO{k^3 \log{n} / \eps^4}$
  time, the data structure either (a) returns the sink component of an
  $(r,t)$-cut of weight at most $\ec_{t,k} + \eps \lambda$, or (b)
  declares that $\ec_{t,k} > \ec$.
\end{lemma}

\begin{proof}
  We first apply \reflemma{apx-rec-sparsification} to $G$ with root
  $r$ and parameters $\lambda$, $k$, and $c \eps$ for a sufficiently
  small constant $c > 0$.  This produces an edge capacitated graph
  $\sparseG= (\sparseV,\sparseE)$, where $r \in \sparseV$ and
  $\sparseV \subseteq V$.  We briefly highlight the features of
  $\sparseG$ guaranteed by \reflemma{apx-rec-sparsification} that we
  leverage.  The edge weights in $\sparseG$ are scaled down so that
  the weight $\ec$ in $G$ corresponds to weight
  $\bigO{k \log{n}/ \eps^2}$ in $\sparseG$.  The edge weights are
  integral, with value between $1$ and $\bigO{k \log{n} / \eps^2}$.
  Every vertex has unweighted in-degree at most
  $\bigO{k \log{n} / \eps^2}$.  Lastly, for every non-root vertex
  $v \in \sparseV - r$, there is an edge from $r$ to $v$ with capacity
  at least $\bigOmega{\log{n}/\eps}$.

  With high probability, we have the following guarantees on the cuts
  of $\sparseG$.  Modulo scaling, every $r$-cut in $\sparseG$ has
  weight no less than the minimum of its weight in $G$ and $2 \ec$.
  Additionally, moduling scaling, the sink component of an $r$-cut in
  $G$ with capacity at most $\ec$ and at most $k$ vertices in the sink
  component is preserved in $\sparseV$, and the corresponding cut in
  $\sparseG$ has weight at most an $c_0 \eps \ec$ additive factor
  bigger in $\sparseG$, for any desired constant $c_0 > 0$. In
  particular we preserve $\ec_{t,k}$ within the desired approximation
  factor for all $t$ such that $\ec_{t,k} \leq \ec$. Henceforth we
  assume that the edge cuts are preserved in the sense described
  above. Otherwise we consider the algorithm to have failed.

  We propose a data structure that, given $t \in V$, will search for a
  small $(r,t)$-cut in $\sparseG$ via a customized, edge-capacitated
  flow algorithm. The search may or may not return the sink component
  of an $(r,t)$-cut.  If the search does return a sink component, and
  the corresponding in-cut in $\sparseG$ has weight that, upon
  rescaling back to the scale of the input graph $G$, is at most
  $(1 + \eps / 2) \ec$, the data structure returns it. Otherwise the
  data structure indicates that $\ec_{t,k} > \ec$.

  To develop the $(r,t)$-cut algorithm, let $\revG$ be the reversed
  graph of $\sparseG$.  In $\revG$, given $t \in V$, we run a
  specialization of the Ford-Fulkerson algorithm \cite{ford-fulkerson}
  with source $t$ and sink $r$ that either computes a minimum
  $(t,r)$-cut or concludes that the minimum $(t,r)$-cut is at least
  $\bigO{k \log{n} / \eps^2}$ after $\bigO{k \log{n} / \eps^2}$
  iterations.  To briefly review, each iteration in the Ford-Fulkerson
  algorithm searches for a path from $t$ to $r$ in the residual graph
  of the flow to that point.  If such a path is found, then it routes
  one unit of flow along this path, and updates the residual graph by
  reversing (one unit capacity) of each edge along the path.  After
  $\ell$ successful iterations we have a flow of size $\ell$ and in
  particular the minimum $(t,r)$-cut is at least $\ell$.  If, after
  $\ell$ iterations, there is no path in the residual graph from $t$
  to $r$, then the set of vertices reachable from $t$ gives a minimum
  $(t,r)$-cut of size $\ell$.

  Within the Ford-Fulkerson framework, we give a refined analysis that
  takes advantages of the auxiliary $(v,r)$ edges (for all $v \neq r$)
  that each have capacity at least $\bigOmega{\log{n} / \eps}$.  Call
  a non-root vertex $v$ \defterm{saturated} if the auxiliary edge
  $(v,r)$ is saturated; that is, if $(v,r)$ is not in the residual
  graph.  (A vertex $v$ is called \defterm{unsaturated} if it is not
  saturated.)  We modify the search for an augmenting path so that
  whenever we visit an unsaturated $v$, we automatically complete a
  path to $r$ via $(v,r)$.  It remains to bound the running time of
  this search.  We first bound the number of saturated vertices.
  \begin{claims}
  \item \labelclaim{saturated-vertices} There are at most
    $\bigO{k / \eps}$ saturated $v$'s.
  \end{claims}
  Indeed, each saturated $v$ implies $\bigOmega{\log{n} / \eps}$
  units of flow via the edge $(v,r)$. The size of the flow is
  limited to $\bigO{k \log{n} / \eps^2}$.

  The above bound on the number of saturated vertices leads to
  the following bound on the total number of edges visited in each
  search.
  \begin{claims}
  \item \labelclaim{local-rec-search-length} Every (modified) search
    for an augmenting path traverses at most
    $\bigO{k^2 / \eps^2}$ edges.
  \end{claims}

  We first observe that every vertex visited in the search, except the
  unsaturated vertex terminating the search, is a saturated vertex.
  By \refclaim{saturated-vertices}, there are at most
  $\bigO{k / \eps}$ saturated vertices.  In turn there are at most
  $\bigO{k^2 / \eps^2}$ edges between saturated vertices.  Thus we can
  traverse at most $\bigO{k^2 / \eps^2}$ edges before visiting either
  an unsaturated vertex or $r$, as claimed.

  \refclaim{local-rec-search-length} implies that each iteration takes
  $\bigO{k^2/\eps^2}$ time.  The algorithm runs for at most
  $\bigO{k \log{n} / \eps^2}$ iterations before either finding a
  minimum $(r,t)$-cut or concluding that the minimum $(r,t)$-cut in
  $\sparseG$ is at least $\bigO{k \log{n} / \eps^2}$ (which
  corresponds to weight $\bigO{\lambda}$ in $G$).  The total running
  time follows.
\end{proof}

We now present the overall algorithm for finding $r$-cuts with small
sink components. The algorithm combines \reflemma{local-apx-rec} with
randomly sampling for a vertex $t$ in the sink component of the
desired $r$-cut.

\begin{lemma}
  \labellemma{apx-rec-small-sink} Let $G = (V,E)$ be a directed graph
  with positive edge weights $w: E \to \preals$. Let $r \in V$ be a
  fixed root vertex.  Let $\ec > 0$ and $k> 0$ be given parameters.
  There is a randomized algorithm that runs in
  $\bigO{m \log k + n k^2 \log{n}^2 / \eps^4}$ time and has the
  following guarantee.  If there is an $r$-cut of capacity at most
  $\ec$ and where the sink component has at most $k$ vertices, then
  with high probability, the algorithm returns an $r$-cut of capacity
  at most $\epsmore \ec$.
\end{lemma}

\begin{proof}
  Let $\optSink$ be the sink component of the minimum $r$-cut subject
  to $\sizeof{\optSink} \leq k$. Suppose the capacity of the in-cut of
  $\optSink$ is at most $\ec$. (Otherwise the algorithm makes no
  guarantee.)

  Suppose we had a factor-2 overestimate $\ell$ of the number of
  vertices in $\optSink$; i.e.,
  $\sizeof{\optSink} \leq \ell \leq 2 \sizeof{\optSink}$. We apply
  \reflemma{local-apx-rec} with parameter upper bound $\ec$ on the
  size of the cut and $\ell$ on the number of vertices in the sink
  component, which produces a deterministic data structure that with
  high probability is correct for all queries. Let us assume the data
  structure is correct (and otherwise the algorithm fails). We then
  randomly sample $\bigO{n \log{n} / \ell}$ vertices from $V - r$. For
  each sampled vertex $t$, we query the data structure from
  \reflemma{local-apx-rec}. Observe that if $t \in \optSink$, then the
  query for $t$ return an $(r,t)$-cut with capacity at most
  $\epsmore \ec$. With high probability we sample at least one vertex
  from $\optSink$, with produces the desired $r$-cut. By
  \reflemma{local-apx-rec}, the running time for
  $\bigO{n \log{n} / \ell}$ queries is
  \begin{math}
    \bigO{m + n \ell^2 \log{n}^2 / \eps^4}.
  \end{math}

  While we do not have such an estimate $\ell$ \emph{a priori}, we can
  try all powers of $2$ between $1$ and $2k$. One of these choices of
  $\ell$ will be accurate and succeed with high probability.  Note
  that the sum of $\bigO{n \ell^2 \log{n}^2 / \eps^4}$ over the range
  of $\ell$ is dominated by the maximum $\ell$.  The claimed running
  time follows.
\end{proof}


\paragraph{Rooted connectivity for large sink components.}
The second subroutine we present is better suited for cases where the
sink component is very large.
\begin{lemma}
  \labellemma{apx-rec-big-sink} Let $G = (V,E)$ be a directed graph
  with positive edge weights.  Let $r \in V$ be a fixed root.  Let
  $\eps, k, \ec > 0$ be given parameters with $\eps$ sufficiently
  small.  Let $\optec$ be the minimum weight of all $r$-cuts where the
  sink component has between $k/2$ and $k$ vertices.  Then there is a
  randomized $\apxO{n^2 / \eps^2 + n^{2.5} / k}$ time algorithm that
  has the following guarantee. If $\optec \leq \ec$, then with
  probability, the algorithm returns an $r$-cut of capacity at most
  $\optec + \eps \ec$.
\end{lemma}

\begin{proof}
  We assume the graph $\sparseG = (\sparseV, \sparseE)$ produced by
  \reflemma{apx-rec-sparsification} with parameters $k$, $\ec$, and
  $c \eps$ for a sufficiently small constant $c > 0$.  Let
  $\optsink \subseteq V - r$ be the sink component of the minimum
  $r$-cut subject to $k/2 \leq \sizeof{\optsink} \leq k$, and suppose
  the capacity of $\optsink$ is at most $\ec$.  We randomly sample
  $\bigO{n \log{n} / k}$ sinks $t \in \sparseV$, and for each compute
  the minimum $(r,t)$ cut. We output the minimum of these cuts.

  If $\optec \leq \ec$, then \reflemma{apx-rec-sparsification} asserts
  that with high probability we have $\optsink \subseteq \sparseV$.
  With high probability, some $t$ will be drawn from $\optsink$ and
  the corresponding $(r,t)$-cut has capacity at most
  $\epsmore \optec + \eps \lambda$.  We use
  $\ectime{m}{n} = \apxO{m+n^{1.5}}$ \cite{brand+21}.  By
  \reflemma{apx-rec-sparsification}, we have
  $m = \bigO{n k \log{n} / \eps^2}$.
\end{proof}

\paragraph{Rooted and global edge connectivity.}
Finally we combine the two approaches above for rooted connectivity --
\reflemma{apx-rec-small-sink} for small components, and
\reflemma{apx-rec-big-sink} for large components -- in an overall
algorithm for rooted edge connectivity and establish
\reftheorem{apx-ec}. We restate \reftheorem{apx-ec} for the sake of
convenience.

\ApxEC*
\begin{proof}
  We focus on rooted connectivity from which global connectivity
  follows immediately.  As the weights are polynomially bounded, we
  can guess the rooted connectivity $\lambda$ and the number of
  vertices $k$ in the sink component of the minimum rooted cut up to a
  multiplicative factor of 2 with polylogarithmic overhead. For each
  choice of $k$, we apply the faster of \reflemma{apx-rec-small-sink}
  and \reflemma{apx-rec-big-sink}.  We balance the
  \begin{math}
    \apxO{n  k^2 / \eps^4}
  \end{math}
  running time of \reflemma{apx-rec-small-sink} with the
  \begin{math}
    \apxO{n^2 / \eps^2 + n^{2.5} / k}
  \end{math}
  running time of \reflemma{apx-rec-big-sink}.  For
  $k = \sqrt{n} \eps^{4/3}$, we obtain the claimed running time.
\end{proof}

\paragraph{$o(mn)$-time approximations.}

\providecommand{\optSink}{T^{\star}}
\providecommand{\sparseG}{G_0}  %

Recall from the introduction that there is theoretical interest
in obtaining a $o(mn)$ running time, partly to improve on the
running time of \cite{hao-orlin} for exact edge connectivity and
partly due to the connection to the flow-decomposition barrier.
\reftheorem{apx-ec} gives an $o(mn)$ running time for all but
the extremely sparse regime where $m = n^{1+o(1)}$.
Part of the problem is that the $\apxO{m + n^{1.5}}$ running
time of \cite{brand+21} is not as compelling for extremely sparse
graphs.
However, the minimum $(s,t)$-cut can also be obtained in
$\apxO{m^{1.5 - \delta}}$ time for $\delta = 1/128$
\cite{glp-21}. This is slightly faster than $\apxO{m + n^{1.5}}$ for
$m = n^{1+o(1)}$.
Using this second $(s,t)$-flow algorithm leads to the following
slightly improved running time for extremely sparse graphs.

\begin{lemma}
  \labellemma{sparse-apx-rec} Let $\eps \in (0,1)$, and let
  $G = (V,E)$ be a directed graph with polynomially bounded edge
  weights.  Let $\eps > 0$.  Suppose that the minimum $(s,t)$-cut can
  be computed in $\apxO{m^{3/2-\delta}}$ time for a constant
  $\delta > 0$.  A $\epsmore$-approximate minimum rooted (or global)
  cut in $G$ can be computed with high probability in
  $\apxO{n m^{1 - 2\delta/3} / \eps^{4/3}}$ randomized time.
\end{lemma}
\begin{proof}
  We take the same approach as \reftheorem{apx-ec} except we use the
  $\apxO{m^{3/2-\delta}}$ running time for $(s,t)$-flow, and we don't
  sparsify the graph first.  As before, we guess the weight $\lambda$
  of the minimum $r$-cut and the number of vertices $k$ in the sink
  component of the minimum $r$-cut to within a multiplicative factor
  of 2 with polylogarithmic overhead.  For a fixed choice of $k$, we
  run the faster of two options.  The first option is to invoke
  \reflemma{apx-rec-small-sink} which runs in
  $\bigO{n k^2 \log{n}^2 / \eps^4}$ randomized time.  The second
  option replaces \reflemma{apx-rec-big-sink} and is as follows. We
  sample $\bigO{k \log{n}}$ vertices with high probability. For each
  sampled vertex $t$, we compute the minimum $(r,t)$-cut in
  $\apxO{m^{3/2-\delta}}$ time. For a correct value of $k$, with high
  probability, one of the sampled vertices $t$ is in the sink
  component of the minimum cut, and we obtain the minimum $(r,t)$-cut.
  The total running time of this combined approach is obtained (up to
  logarithmic factors) by choosing $k$ to balance the two running
  times of $\apxO{n k^2 \log{n} / \eps^4}$ and
  $\apxO{n m^{3/2 - \delta}/ k}$. For
  $k = \eps^{4/3} m^{1/2- \delta / 3}$, we obtain the claimed running
  time.
\end{proof}
Balancing \reflemma{sparse-apx-rec} with \reftheorem{apx-ec} gives the
$\apxO{(mn)^{1-c} / \eps^2}$ running time, where $c> 0$ is a constant,
that is claimed in \refcorollary{flow-barrier-apx-ec}.



\section{Rooted vertex connectivity}

\labelsection{apx-vc}

In this section we present the approximation algorithm for rooted and
global vertex connectivity. Similar to \refsection{apx-ec}, the main
focus is on rooted connectivity, and the algorithm is presented in
three main parts. All three parts are parameterized by values
$\vc > 0$ and $k \in \naturalnumbers$ that, in principle, are meant to
be constant factor estimates for the rooted vertex connectivity and
the number of vertices in the sink component of the minimum rooted
vertex cut. The first part reduces the number of edges to roughly $nk$
and the rooted connectivity to $k$ in a graph with integer
weights. This sparsification is used in the remaining two parts. The
second part gives a roughly $nk^2$ time approximation algorithm for
the minimum rooted cut. The third part gives a roughly
$n^2 + n^{2.5} / k$ time approximation algorithm for the minimum
rooted cut. Balancing term leads to the claimed running time. The
rooted connectivity algorithm then leads to a global connectivity
algorithm via an argument due to \cite{hrg-00} (with some
modifications).

\paragraph{Sparsification.}
The first part is a sparsification lemma that preserves rooted vertex
cuts where the number of vertices in the sink component is below some
given parameter.  In the following, we let $\outneighbors{r}{G}$
denote the set of out-neighbors of $r$ in the graph $G$. We omit $G$
and simply write $\outneighbors{r}$ when $G$ can be inferred from the
context.
\begin{lemma}
  \labellemma{apx-rvc-sparsification} Let $G = (V,E)$ be a directed
  graph with positive vertex weights.  Let $r \in V$ be a fixed
  vertex.  Let $k, \vc > 0$ be given parameters.  Let
  $V' = V \setminus \parof{\setof{r} \cup \outneighbors{r}}$.  In
  randomized linear time, one can compute a randomized directed and
  vertex-weighted graph $\sparseG = (\sparseV, \sparseE)$, and a
  scaling factor $\tau > 0$, with the following properties.
  \begin{properties}
  \item \label{rvcs-root} $r \in \sparseV$.
  \item \label{rvcs-sinks} Let
    $\sparseV' = \sparseV \setminus \parof{\setof{r} \cup
      \outneighbors{r}{\sparseG}}$. We have $\sparseV' = V'$.
  \item \label{rvcs-weights} $\sparseG$ has integer vertex weights
    between $0$ and
    $\bigO{k \log{n} / \eps^2}$.
  \item \label{rvcs-degree} Every vertex $v \in \sparseV$ has at most
    $\bigO{k \log{n} / \eps^2}$ incoming edges.
  \item \label{rvcs-zero} Every vertex $v$ with weight $0$ has no
    outgoing edges.
  \item \label{rvcs-cut-weights} With high probability, for all
    $S \subseteq V'$, the weight of the vertex in-cut induced by $S$
    in $\sparseG$ (up to scaling by $\tau$) is at least the minimum of
    the $\epsless$ times the weight of the induced vertex in-cut in
    $G$ or $c \vc$ (for any desired constant $c > 1$), and at most
    $\epsmore$ times its weight in $G$ plus $\eps \vc \sizeof{S} / k$.
  \item \label{rvcs-small-cut-sets} With high probability, for all
    $S \subseteq V'$ such that $\sizeof{S} \leq k$ and the weight of
    the induced vertex in-cut is $\leq \bigO{\vc}$, we have
    $S \subseteq \sparseV'$. (That is, $S$ is still the sink component
    of an $r$-cut in $\sparseG$.)
  \end{properties}
  In particular, if the minimum vertex $r$-cut has weight
  $\bigTheta{\vc}$, and the sink component of a minimum vertex $r$-cut
  has at most $k$ vertices, then with high probability $\sparseG$
  preserves the minimum vertex $r$-cut up to a
  $\apxmore$-multiplicative factor.
\end{lemma}
\begin{proof}
  Consider the following randomized algorithm applied to the input
  graph $G$.
  \begin{quote}
    \begin{steps}
    \item \labelstep{rvc-first-step} Let
      $\tau = c_{\tau} \eps^2 \vc / k \log{n}$ and let
      $\Delta = c_{\Delta} k \log{n} / \eps^2$, where $c_{\tau} > 0$
      is a sufficiently small constant and $c_{\Delta} > 0$ is a
      sufficiently large constant.
    \item \labelstep{rvc-importance-sample} Important sample each vertex
      weight to be a discrete multiple of $\tau$.
    \item \labelstep{rvc-extra-vertices} For each vertex $v$,
      introduce an auxiliary vertex $a_v$ with weight
      $\eps \vc / 2 k$. Add edges from the $r$ to $a_v$, and from
      $a_v$ to $v$.  %
      \commentitem{Decreasing $c_{\tau}$ and $\eps$ as needed, we
        assume that $\vc$ and $\eps \vc / 2k$ are multiples of
        $\tau$}.
    \item Remove all outgoing edges from any vertex with weight $0$.
    \item \labelstep{rvc-rescale} Scale down all vertex weights by
      $\tau$ (which makes them integers).
    \item \labelstep{rvc-truncate} Truncate all vertex weights to be
      at most $c_{w} k \log{n} / \eps^2$ for a sufficiently large
      constant $c_w > 0$.
    \item \labelstep{rvc-contract} For all $v$ with unweighted
      in-degree at least $\Delta$, replace all incoming edges to $v$
      with a single edge from $r$.
    \end{steps}
  \end{quote}
  Let $\sparseG$ be the graph obtained at the end of the algorithm
  above. Properties \ref{rvcs-root} through \ref{rvcs-zero} follow
  directly from the construction. The remaining proof is dedicated to
  proving the high probability bounds of \ref{rvcs-cut-weights} and
  \ref{rvcs-small-cut-sets}. We first show that the initial steps
  \refstep{rvc-first-step} to \refstep{rvc-extra-vertices} -- before
  rescaling -- preserves the minimum weight rooted vertex-cut
  approximately in the sense of \ref{rvcs-cut-weights} (sans
  scaling). We then analyze the remaining steps.  For each set $S$,
  let $f(S)$ denote the weight of the vertex in-cut of $S$. Let $g(S)$
  denote the randomized weight of the vertex in-cut after step
  \refstep{rvc-importance-sample}. Let $h(S)$ denote the randomized
  weight of the vertex in-cut after step \refstep{rvc-extra-vertices}.
  \begin{claims}
  \item
    \labelclaim{rvc-concentration}
    With high probability, for all $S \subseteq V$, we have
    \begin{align*}
      \absvof{g(S) - f(S)} \leq \eps \f{S} + \frac{\eps \vc \sizeof{S}}{2k}.
    \end{align*}
  \end{claims}
  The claim and proof are similar to \refclaim{rec-concentration} in
  the proof of \reflemma{apx-rec-sparsification}. The claim consists
  of an upper bound and a lower bound on $g(S)$ for all $S$ and we
  first show the lower bound holds with high probability.  Fix any set
  $S$.  $g(S)$ is an independent sum with expected value $f(S)$ and
  where each term in the sum is nonnegative and varies by at most
  $\tau$. Concentration bounds (see footnote
  \ref{footnote:additive-chernoff} on page
  \pageref{footnote:additive-chernoff}) imply that for any
  $\gamma \geq 0$, we have
  \begin{align*}
    \probof{g(S) \leq \epsless  f(S) - \gamma} \leq e^{- \eps \gamma /
    \tau}
    =                           %
    n^{- \gamma k \log{n} / c_{\tau} \eps \vc},
  \end{align*}
  In particular, for
  $\gamma = \eps \vc \sizeof{S} / 2 k$, the RHS is at most
  \begin{math}
    n^{- c_0 \sizeof{S}}
  \end{math}
  where $c_0 > 0$ is again a constant under our control (via
  $c_{\tau}$). For sufficiently large $c_0$, we can take the union
  bound over all sets $S$, establishing the high probability lower
  bound. The high probability upper bound follows by a symmetric
  argument.

  Now we analyze the vertex $r$-cuts after step
  \refstep{rvc-extra-vertices}.
  Recall that for $S \subseteq V$, $h(S)$ denotes the weight
  of the in-cut of $S$ after adding auxiliary vertices in step
  \refstep{rvc-extra-vertices}.
  \begin{claims}
  \item \labelclaim{rvc-padded-concentration} For all
    $S \subseteq V'$, we have
    \begin{align*}
      \epsless \f{S} \leq h(S) \leq \epsmore \f{S} + \frac{\eps \vc\sizeof{S}}{k}.
    \end{align*}
  \end{claims}
  This claim and its proof is similar to
  \refclaim{rec-padded-concentration} in
  \reflemma{apx-rec-sparsification}.  We have
  $h(S) = g(S) + \eps \vc \sizeof{S} / 2 k$ for all
  $S \subseteq V'$.  The additive factor of
  $\eps \vc \sizeof{S} / 2 k$ combine with the high-probability
  additive error in \refclaim{rec-concentration} to establish the
  claim.

  We point out that \refclaim{rvc-padded-concentration} implies that,
  with high probability after step \refstep{rvc-extra-vertices}, the
  weights of all the vertex $r$-cuts are preserved the approximate
  sense described by \ref{rvcs-cut-weights} (without the scaling).
  Henceforth we assume that the high probability event of
  \refclaim{rvc-padded-concentration} holds.  Now, after step
  \refstep{rvc-extra-vertices}, all the weights are integer multiples
  of $\tau$. We have $\vc / \tau = \bigO{k \log{n} / \eps^2}$.  After
  scaling down by $\tau$ in step \refstep{rvc-rescale}, we still
  preserve the $r$-cuts per property
  \ref{rvcs-cut-weights}. Truncating weights in $\sparseG$ to
  $\bigO{\vc / \tau}$ decreases the weight of some cuts, but to no
  less than $\bigO{\vc / \tau}$ (which maps to $\bigO{\vc}$ when
  rescaled back to the scale of $G$). Removing the outgoing edges of
  vertices with weight $0$ also has no impact on the weight of any
  vertex $r$-cut. The final step adding edges from $r$ only eliminates
  some of the vertex $r$-cuts from consideration and does not impact
  the weight of the remaining vertex $r$-cuts. This establishes
  \ref{rvcs-cut-weights}.

  It remains to prove \ref{rvcs-small-cut-sets} and in particular we
  must show that it is not impacted by the last step,
  \refstep{rvc-contract}.  Recall that step \refstep{rvc-contract}
  replaces the incoming edges to any vertex $v$ with unweighted
  in-degree greater than $\Delta = \bigO{k \log{n} / \eps^2}$ with a
  single edge in $r$.  In particular, this edge places $v$ in
  $\outneighbors{r}$ and destroys all $r$-cuts where the sink
  component contains $v$.  Let $\Sink \subset V'$ be the sink
  component of a vertex $r$-cut in $G$ where the capacity of the cut
  is at most $\vc$, and $\sizeof{\Sink} \leq k$.  We want to show that
  all vertices in $\Sink$ have in-degree less than $\Delta$, in which
  case the extra edges in \refstep{rvc-contract} have no impact on
  $\Sink$.  The vertex in-cut induced by $\Sink$ has weight at most
  $(1 + 2 \eps) \vc / \tau$ in the randomized graph before
  \refstep{rvc-contract} (per \ref{rvcs-cut-weights}).  Fix any
  $v \in \Sink$ and consider the edges going into $v$. At most $k - 1$
  of those edges can come from another vertex in $\Sink$, since
  $\Sink$ has at most $k$ vertices. The remaining edges must be from
  vertices in the vertex in-cut of $\Sink$.  Each of these vertices
  have weight at least $1$, and by \ref{rvcs-cut-weights} the in-cut
  has weight at most $(1 + 2 \eps) \vc / \tau$, so there are at most
  $\bigO{\vc / \tau}$ of these vertices.  This gives a maximum total
  of less than $\Delta$ edges incident to $v$, as desired.  In
  conclusion, for any vertex $v$ that is the endpoint to at least
  $\Delta$ edges, it is safe to replace all of $v$'s incoming edges
  with a single edge from the root, without violating
  \ref{rvcs-small-cut-sets}. This establishes
  \ref{rvcs-small-cut-sets} and completes the proof.
\end{proof}

\paragraph{Rooted vertex connectivity for small sink components.}
This section presents an approximation algorithm for rooted vertex
connectivity for the particular setting where the sink component of
the minimum rooted cut is small.  The algorithm is similar to the
algorithms presented in \reflemma{local-apx-rec} and
\reflemma{apx-rec-small-sink} for rooted edge connectivity, and has
the same inspirations (from \cite{forster+20,cq-simple-connectivity})
and motivations. As with \reflemma{apx-rec-small-sink}, the local
algorithm presented here is customized to take full advantage of the
properties of the graph produced by the sparsification lemma,
\reflemma{apx-rvc-sparsification}, and is a deterministic algorithm
with better dependency on $k$ and $\eps$ compared to previous
algorithms.

\begin{lemma}
  \labellemma{local-apx-rvc} Let $G = (V,E)$ be a directed graph with
  positive vertex weights.  Let $r \in V$ be a fixed root vertex.  Let
  $\eps \in (0,1)$, $\vc > 0$ and $k \in \naturalnumbers$ be given
  parameters.  There is a randomized linear time Monte Carlo algorithm
  that, with high probability, produces a deterministic data structure
  that supports the following query.

  For
  $t \in V' \defeq V \setminus \parof{\setof{r} \cup
    \outneighbors{r}}$, let $\vc_{t,k}$ denote the weight of the
  minimum $(r,t)$-vertex cut such that the sink component has at most
  $k$ vertices. Given $t \in V'$, deterministically in
  $\bigO{k^3 \log{n}^2 / \eps^4}$ time, the data structure either (a)
  returns the sink component of a minimum $(r,t)$-vertex cut of weight
  at most $(1 + \eps) \vc_{t,k}$, or (b) declares that
  $\vc_{t,k} > \vc$.
\end{lemma}

\begin{proof}
  We first apply \reflemma{apx-rvc-sparsification} to $G$ with root
  $r$ and parameters $\vc$, $k$, and $c \eps$ for a constant $c > 0$
  sufficiently small. This produces a vertex capacitated graph
  $\sparseG= (V_0,E_0)$ with $V \subset V_0$.  We highlight the
  features that we leverage.  All new vertices (in $V_0 \setminus V$)
  are in $\outneighbors{r}{\sparseG}$; that is, $V'$ equals
  \begin{math}
    \sparseV' \defeq \sparseV \setminus \parof{\setof{r} \cup
      \outneighbors{r}{\sparseG}}.
  \end{math}
  Put alternatively, none of the new vertices are in the sink
  component of any $r$-cut.  The vertex weights are integers between
  $0$ and $\bigO{k \log{n} / \eps^2}$.  Every vertex has unweighted
  in-degree at most $\bigO{k \log{n} / \eps^2}$.  Every vertex with
  weight $0$ has no outgoing edges.

  With high probability, we have the following guarantees on the
  vertex $r$-cuts of $\sparseG$.  The vertex weights in $\sparseG$ are
  scaled so that a weight of $\vc$ in $G$ corresponds to weight
  $\bigO{k \log{n}/ \eps^2}$ in $\sparseG$.  Modulo scaling, every
  vertex $r$-cut in $\sparseG$ has weight no less than the minimum of
  its weight in $G$ and $2 \vc$.  Additionally, modulo scaling, for
  every vertex $r$-cut in $G$ with capacity at most $\vc$ and at most
  $k$ vertices in the sink component, the corresponding vertex cut in
  $\sparseG$ has weight at most a $c_0 \eps \vc$ additive factor
  larger than in $G$, for any desired constant $c_0 > 0$.  We consider
  the algorithm to fail if the cuts are not preserved in the sense
  described above.

  Given $t \in V$, the data structure will search for a small
  $(r,t)$-cut in $\sparseG$ via a customized, edge-capacitated flow
  algorithm. This algorithm may or may not return the sink component
  of $(r,t)$-cut. If the search does return a sink component, and the
  corresponding vertex in-cut in $\sparseG$ has weight that, upon
  rescaling back to the scale of the input graph $G$, is at most
  $(1 + \eps / 2) \vc$, the data structure returns it. Otherwise the
  data structure indicates that $\vc_{t,k} > \vc$.

  Proceeding with the flow algorithm, let $\revG$ be the reverse of
  $\sparseG$, and let $\splitG$ be the standard ``split-graph'' of
  $\revG$ modeling vertex capacities with edge capacities.  We recall
  that the split graph splits each vertex $v$ into an auxiliary
  ``in-vertex'' $\vin$ and an auxiliary ``out-vertex'' $\vout$.  For
  each $v$ there is a new edge $(\vin,\vout)$ with capacity equal to
  the vertex capacity of $v$.  Each edge $(u,v)$ is replaced with an
  edge $(\uout,\vin)$ with capacity\footnote{Usually, this edge is set
    to capacity $\infty$, but either the weight of $u$ or the weight
    of $v$ are also valid.}  equal to the vertex capacity of $u$.
  Every $(r,t)$-vertex cut in $\sparseG$ maps to a $(\tout,\rin)$-edge
  cut in $\revG$ with the same capacity.  Any $(\tout,\rin)$-edge
  capacitated cut maps to a $(r,t)$-vertex cut in $\sparseG$ (with
  negligible overhead in the running time). Now, recall that for each
  $v \in V'$, the sparsification procedures introduces an auxiliary
  path $(r,a_v,\rin)$ where $a_v$ is was given weight
  $\bigTheta{\eps \vc / k}$. It is convenient to replace the
  corresponding auxiliary path $(\vout, a_v^-, a_v^+, \rin)$ in
  $\revG$ with a single edge $(\vout,\rin)$ with capacity equal to the
  weight of $a_v$.  This does not effective the weight of the minimum
  $(\tout,\rin)$-edge cut for any $t \in V'$.  This adjustment can be
  easily made within the allotted preprocessing time.

  In this graph, given $t \in V'$, we run a specialization of the
  Ford-Fulkerson algorithm \cite{ford-fulkerson} that either computes
  a minimum $(\tout,\rin)$-cut or concludes that the minimum
  $(\tout,\rin)$-cut is at least $\bigO{k \log{n} / \eps^2}$ (which
  corresponds to $\bigO{\vc}$ in $G$) after
  $\bigO{k \log{n} / \eps^2}$ iterations.  To briefly review, each
  iteration in the Ford-Fulkerson algorithm searches for a path from
  $t$ to $r$ in the residual graph. If such a path is found, then it
  routes one unit of flow along this path, and updates the residual
  graph by reversing (one unit capacity) of each edge along the path.
  After $\ell$ successful iterations we have a flow of size $\ell$.
  If, after $\ell$ iterations, there is no path in the residual graph
  from $\tout$ to $\rin$, then the set of vertices reachable from $t$
  gives a minimum $(\tout,\rin)$-cut of size $\ell$.  Observe that
  updating the residual graph along a $(\tout,\rin)$-path preserves
  the weighted in-degree and out-degree of every vertex except $\tout$
  and $\rin$. The weighted out-degree of $\tout$ decreases by 1 and
  the weighted in-degree of $\rin$ changes by $1$. Moreover, updating
  the residual graph along a path increases the unweighted out-degree
  of any vertex by at most one, since a path contains at most one edge
  going into any single vertex.  Since every vertex initially has
  unweighted out-degree at most $\bigO{k \log{n} / \eps^2}$ in $\revG$
  (reversing the upper bound on the unweighted in-degrees in
  $\sparseG$), and the flow algorithm updates the residual graph along
  at most $\bigO{k \log{n} / \eps^2}$ paths before terminating, the
  maximum unweighted out-degree over all vertices never exceeds
  $\bigO{k \log{n} / \eps^2}$.

  Within the Ford-Fulkerson framework, we give a refined analysis that
  takes advantages of the auxiliary $(\vout,\rin)$ edges.  Call an
  out-vertex $\vout$ \defterm{saturated} if the auxiliary edge
  $(\vout,\rin)$ is saturated; that is, if $(\vout,\rin)$ is not in
  the residual graph.  Call an in-vertex $\vin$ \defterm{saturated} if
  the edge $(\vin,\vout)$ is saturated and $\vout$ is not saturated.
  (A vertex $\vout$ or $\vin$ is called \defterm{unsaturated} if it is
  not saturated.)  We modify the search for an augmenting path to
  effectively end when we first visit an unsaturated vertex $\vout$ or
  an unsaturated $\vin$.  If we visit an unsaturated $\vin$, then we
  automatically complete a path to $\rin$ via $\vout$.  If we find an
  unsaturated $\vout$, then we automatically complete a path to $\rin$
  via the edge $(\vout,\rin)$. It remains to bound the running time of
  this search. We first bound the number of saturated $\vout$'s.

  \begin{claims}
  \item \labelclaim{saturated-vouts} There are at most
    $\bigO{k / \eps}$ saturated $\vout$'s.
  \end{claims}
  Indeed, each saturated $\vout$ implies $\bigO{\log{n} / \eps}$ units
  of flow along $(\vout,\rin)$, and the flow is bounded above
  $\bigO{k \log{n} / \eps^2}$.

  Note that \refclaim{saturated-vouts} also implies there are at most
  $\bigO{k / \eps}$ $\vin$'s such that $\vout$ is saturated. The next
  claim bounds the total out-degree of saturated $\vin$'s.

  \begin{claims}
  \item \labelclaim{saturated-vin-out-degrees} The sum of out-degrees
    of saturated $\vin$'s is at most the amount of flow that has been
    routed to $\rin$.
  \end{claims}

  Indeed, the out-degree of a $\vin$ in the residual graph is bounded
  above by the amount of flow through $(\vin,\vout)$, since initially
  $(\vin,\vout)$ is the only outgoing edge from $\vin$.  Recall that
  if $\vin$ is saturated, then by definition $\vout$ is unsaturated.
  As long as $\vout$ is unsaturated, each unit of flow through
  $(\vin,\vout)$ goes directly to $\rin$ via the edge $(\vout,\rin)$,
  and can be charged to the total flow.

  We now apply the above two claims to bound the total running time
  for each search, as follows.
  \begin{claims}
  \item \labelclaim{local-rvc-search-length} Every (modified) search
    for an augmenting path traverses at most
    $\bigO{k^2 \log{n} / \eps^2}$ edges.
  \end{claims}
  We first observe that every vertex visited in the search, except the
  unsaturated vertex terminating the search, is either (a) a saturated
  $\vin$, (b) a saturated $\vout$, or (c) an unsaturated $\vin$ such
  that $\vout$ is saturated.  We will upper bound the number of edges
  traversed in each iteration based on the type of vertex at the
  initial point of that edge.  First, the amount of time spent
  exploring edges leaving (a) a saturated $\vin$ is, by
  \refclaim{saturated-vin-out-degrees}, at most the size of the flow
  at that point, which is at most $\bigO{k \log{n} / \eps^2}$.
  Second, consider the time spent traversing edges leaving either (b)
  a saturated $\vout$ or (c) an unsaturated $\vin$ such that $\vout$
  is saturated. By \refclaim{saturated-vouts}, there are at most
  $\bigO{k / \eps}$ such vertices, and each has out-degree at most
  $\bigO{k \log{n} / \eps^2}$. Thus we spend
  $\bigO{k^2 \log{n} / \eps^2}$ time traversing such edges. All
  together, we obtain an upper bound of $\bigO{k^2 \log{n} / \eps^2}$
  total edges per search.

  \refclaim{local-rvc-search-length} also bounds the running time for
  each iteration. The algorithm runs for at most
  $\bigO{k \log{n} / \eps^2}$ iterations before either finding an
  $(\tout,\rin)$-cut or concluding that the weight of the minimum
  $(\tout,\rin)$-cut, rescaled to the input scale of $G$, is at least
  a constant factor greater than $\vc$.  The total running time
  follows.
\end{proof}

We now present the overall algorithm for finding vertex $r$-cuts with
small sink components. The algorithm combines \reflemma{local-apx-rvc}
with randomly sampling for a vertex $t$ in the sink component of an
approximately minimum $r$-cut.

\begin{lemma}
  \labellemma{apx-rvc-small-sink} Let $G = (V,E)$ be a directed graph
  with positive vertex weights. Let $r \in V$ be a fixed root vertex.
  Let $\eps \in (0,1)$, $\vc > 0$ and $k \in \naturalnumbers$ be given
  parameters.  There is a randomized algorithm that runs in
  $\bigO{m + (n - \outdegree{r}) k^2 \log{n}^3 / \eps^4}$ time and has
  the following guarantee.  If there is a vertex $r$-cut of capacity
  at most $\vc$ and where the sink component has at most $k$ vertices,
  then with high probability, the algorithm returns a vertex
  $(r,t)$-cut of capacity at most $\epsmore \vc$.
\end{lemma}

\begin{proof}
  Let $\optSink$ be the sink component of the minimum vertex $r$-cut
  subject to $\sizeof{\optSink} \leq k$. Assume the capacity of the
  vertex in-cut of $\optSink$ is at most $\vc$ (since otherwise the
  algorithm makes no guarantees). Let
  $V' = V \setminus (\setof{r} \cup \outneighbors{r})$ and note that
  $\sizeof{V'} = n - 1 - \outdegree{r}$.

  Suppose we had a factor-2 overestimate
  $\ell \in \bracketsof{\sizeof{\optSink}, 2 \sizeof{\optSink}}$ of
  the number of vertices in $\optSink$. We apply
  \reflemma{local-apx-rvc} with upper bounds $\vc$ on the size of the
  cut and $\ell$ on the number of vertices in the sink component,
  which returns a data structure that, with high probability, is
  correct for all queries. Let us assume the data structure is correct
  (and otherwise the algorithm fails). We randomly sample
  $\bigO{(n - \outdegree{r}) \log{n} / \ell}$ vertices from $V'$. For
  each sampled vertex $t$, we query the data structure from
  \reflemma{local-apx-rvc}. Observe that if $t \in \optSink$, then the
  query for $t$ returns an $r$-cut with capacity at most
  $\epsmore \vc$. With high probability we sample at least one vertex
  from $\optSink$, which produces the desired $r$-cut. By
  \reflemma{local-apx-rvc}, the total running time to serve all
  queries is
  $\bigO{m + (n - \outdegree{r}) \ell^2 \log{n}^3 / \eps^4}$.

  A factor-2 overestimate $\ell$ can be obtained by enumerating all
  powers of $2$ between $1$ and $2k$. One of these choices of $\ell$
  will be accurate and produce the minimum $r$-cut with high
  probability. Note that the sum of
  $\bigO{(n - \outdegree{r}) \ell^2 \log{n}^3 / \eps^4}$ over this
  range of $\ell$ is dominated by the maximum $\ell$. The claimed
  running time follows.
\end{proof}


\paragraph{Rooted vertex connectivity for large sink components.}

The third and final part (before the overall algorithm) is an
approximation for the rooted vertex cut that is well-suited for large
sink components.
\begin{lemma}
  \labellemma{apx-rvc-big-sink}
  Let $G = (V,E)$ be a directed graph with positive vertex
  weights. Let $r \in V$ be a fixed root vertex. Let $\eps \in (0,1)$,
  $\vc > 0$, and $k \in \naturalnumbers$ be given parameters. There is
  a randomized algorithm that runs in
  \begin{math}
    \apxO{m + (n - \outdegree{r}) \parof{n / \eps^2 + n^{1.5} / k}}
  \end{math}
  time and has the following guarantee. If there is a vertex $r$-cut
  of capacity at most $\vc$ and where the sink component has at most
  $k$ vertices, then with high probability, the algorithm returns a
  vertex $(r,t)$-cut of capacity at most $\epsmore \vc$.
\end{lemma}

\begin{proof}
  Let $\optsink$ be the sink component of the minimum $r$-cut subject
  to $\sizeof{\optsink} \leq k$. We assume the capacity of the $r$-cut
  induced by $\optsink$ is at most $\vc$. (Otherwise the output is not
  well-defined.) Let
  $V' = V \setminus \parof{\setof{r} \cup \outneighbors{r}}$ and note
  that $\sizeof{V'} < n - \outdegree{r}$.

  We apply \reflemma{apx-rvc-sparsification} to produce the graph
  $\sparseG$. \reflemma{apx-rvc-sparsification} succeeds with high
  probability and for the rest of the proof we assume it was
  successful. (Otherwise the algorithm fails.)  We sample
  $\bigO{ \parof{n - \outdegree{r}} \log{n} / k}$ vertices $t \in
  V'$. For each sampled $t$, we compute the minimum $(r,t)$-vertex cut
  in $\sparseG$. With high probability, some $t$ will be drawn from
  the sink component of the true minimum $r$-cut, in which case the
  minimum $(r,t)$-cut in $\sparseG$ gives an $\epsmore$-approximate
  $r$-cut in $G$ (by \reflemma{apx-rvc-sparsification}).  We use
  $\vctime{m}{n} = \apxO{m + n^{1.5}}$ \cite{brand+20}. By
  \reflemma{apx-rvc-sparsification}, we have
  $m = \bigO{n k \log{n} / \eps^2}$. This gives the total running
  time.
\end{proof}

\paragraph{Approximating the rooted and global vertex connectivity.}

We now combine the two parameterized approximation algorithms for
rooted vertex connectivity to give the following overall algorithm for
rooted edge connectivity and establish \reftheorem{apx-rvc}. We
restate \reftheorem{apx-rvc} for the sake of convenience.

\ApxRVC*

\begin{proof}
  The high-level approach is similar to \reftheorem{apx-ec} for edge
  connectivity -- we are balancing two algorithms for rooted vertex
  connectivity, where one is better suited for small sink components,
  the second is better suited for large sink components. Both leverage
  the randomized sparsification lemma.  As before, with
  polylogarithmic overhead, we can assume access to values $\vc$ and
  $k$ that are within a factor 2 of the weight of the minimum $r$-cut
  and the number of vertices in the sink component of the minimum
  $r$-cut, respectively.  For a fixed choice of $k$ and $\vc$, we
  run the faster of two randomized algorithms, both of which would
  succeed with high probability when $k$ and $\vc$ are
  (approximately) correct.  The first option, given by
  \reflemma{apx-rvc-small-sink}, runs in
  $\apxO{(n - \outdegree{r}) k^2 / \eps^3}$.  The second option, given
  by \reflemma{apx-rvc-big-sink}, runs in
  $\apxO{n(n-\outdegree{r}) / \eps^2 + n^{1.5} (n - \outdegree{r}) /
    k}$.  The overall running time is obtained by choosing $k$ to
  balance the running times.  For $k = \eps \sqrt{n}$, we obtain the
  claimed running time.
\end{proof}

Next we use the algorithm for rooted vertex connectivity to obtain an
algorithm for global vertex connectivity and establish
\refcorollary{apx-vc}.  \citet{hrg-00} showed that running times of
the form $(n-\outdegree{r}) T$ for rooted connectivity from a root $r$
imply a randomized $n T$ expected time algorithm for global vertex
connectivity.  \reftheorem{apx-rvc} gives a
$\apxO{m + n (n-\outdegree{r}) / \eps^2}$ running time, so some
modifications have to be made to address the additional $\apxO{m}$
additive factor. We restate \refcorollary{apx-vc} for the sake of
convenience.

\ApxVC*
\begin{proof}
  Let $T = \apxO{n / \eps^2}$.  Let $w: V \to \preals$ denote the
  vertex weights, and let $W = \sum_{v \in V} \weight{v}$ be the total
  weight of the graph. Let $\vc$ denote the weight of the minimum
  global vertex cut. The algorithm samples
  $L = \bigO{W \log{n} / (W - \vc)}$ vertices $r$ in proportion to
  their weight, and -- morally, but not actually -- computes the
  minimum $r$-vertex cut for each sampled vertex $r$ via
  \reftheorem{apx-rvc}. It returns the smallest cut found.

  For the sake of running time, we adjust the algorithm from
  \reftheorem{apx-rvc}. Recall that for a fixed root $r$, and for each
  of a logarithmic number of values for $k$ and $\vc$, the algorithm
  from \reftheorem{apx-rvc} applies \reflemma{apx-rvc-sparsification}
  which reduces the graph to having $\apxO{n k / \eps^2}$ edges and
  rooted connectivity $\apxO{k/\eps^2}$. For fixed $k$ and $\vc$,
  rather than rerun \reflemma{apx-rvc-sparsification} entirely for
  each $r$ we sample, we execute most of it just once for all $r$, and
  make local modifications for each different root $r$. Referring to
  the algorithm given in the proof of
  \reflemma{apx-rvc-sparsification}, observe that the only step that
  directly mentions $r$ is step \refstep{rvc-extra-vertices}, which
  adds auxiliary vertices between $r$ and each other vertex $v$. We
  move this step to the very end of the algorithm. (Here the vertex
  weight of auxiliary vertices is scaled down appropriately.)  It is
  easy to see that the proof of \reflemma{apx-rvc-sparsification}
  still goes through (with minor rearrangement in the argument). The
  advantage is that, over all $L$ roots $r$, we now spend a total of
  $\bigO{m + n L}$ time, rather than $\bigO{m L}$.  Thereafter, the
  rest of the rooted connectivity algorithm takes
  $\apxO{(n-\outdegree{r}) T}$ per root $r$. Note that
  $\apxO{(n-\outdegree{r}) T}$ dominates the $\bigO{n}$ time required
  to complete the sparsification for each root.

  Consider a single root $r$ sampled from $V$ in proportion to its
  weight. The expected running time to compute the minimum $r$-cut is
  \begin{align*}
    \evof{\parof{n - \outdegree{r}} T}
    =                           %
    n T - \frac{T}{W}\sum_{v \in V} \outdegree{v} \weight{v}
    \tago{=}
    n T - \frac{T}{W} \sum_{v \in V} \sum_{x \in \inneighbors{v}} \weight{x}
    \tago{\leq}
    n T \parof{1 - \frac{\vc}{W}}.
  \end{align*}
  Here, in \tagr, $\inneighbors{v}$ denotes the in-neighborhood of
  $v$. The equality is obtained by implicitly interchanging
  sums. \tagr is because for each $v$, the sum
  $\sum_{x \in \inneighbors{v}} \weight{x}$ is the weighted in-degree
  of $v$, and at least $\vc$.  Thus The overall expected running time over
  all the sampled roots is
  \begin{align*}
    \bigO{\prac{W \log{n}}{W - \vc} \evof{\parof{n - \outdegree{r}}
    T}}
    =
    \bigO{\prac{W \log{n}}{W - \vc} \cdot nT \parof{1 -
    \frac{\vc}{W}}}
    =                           %
    \apxO{n^2 / \eps^2}.
  \end{align*}
  Meanwhile, when we sample $\bigO{W \log{n} / (W - \vc)}$ vertices in
  proportion to their weight, then with high probability, at least one
  sampled vertex lies outside the minimum global vertex cut. Such a
  vertex then leads to the minimum global vertex cut with high
  probability.
\end{proof}

\paragraph{$o(mn)$-time approximations for vertex connectivity.}

\providecommand{\optSink}{T^{\star}}
\providecommand{\sparseG}{G_0}  %

We now show how to approximate rooted and global vertex connectivity
in $o(mn)$ time, just as we did previously for edge connectivity. We
point out that the minimum $(s,t)$-edge cut algorithm of \cite{glp-21}
that runs in $\apxO{m^{1.5 - \delta}}$ time for $\delta = 1/128$ is
also an $\apxO{m^{1.5 - \delta}}$ time algorithm for $(s,t)$-vertex
cut by standard reductions. As the ideas here to develop an $o(mn)$
time approximation for vertex connectivity are the same as for edge
connectivity, we restrict ourselves to a sketch.

\begin{lemma}
  \labellemma{sparse-apx-vc} Let $\eps \in (0,1)$, and let $G = (V,E)$
  be a directed graph with polynomially bounded vertex weights.
  Suppose that the minimum $(s,t)$-vertex cut can be computed in
  $\apxO{m^{3/2-\delta}}$ time for a constant $\delta > 0$. For a
  fixed root $r \in V$, a $\epsmore$-approximate minimum vertex
  $r$-cut can be computed with high probability in
  $\apxO{m + (n - \outdegree{r}) m^{1 - 2\delta/3} / \eps^{4/3}}$
  randomized time. A $\epsmore$-approximate minimum global vertex cut
  can be computed with high probability in
  $\apxO{m + n m^{1 - 2 \delta/3} / \eps^{4/3}}$ randomized time.
\end{lemma}
\begin{proof}[Proof sketch]
  We first consider rooted vertex cuts. We take the same approach as
  \reftheorem{apx-rvc} except modifying the algorithm
  \reflemma{apx-rvc-big-sink} as follows. First, we do not sparsify
  the graph. Second, we use the $\apxO{m^{3/2 - \delta}}$ time
  algorithm for $(s,t)$-vertex cut instead of $\apxO{m +
    n^{3/2}}$. The result replaces the running time in
  \reflemma{apx-rvc-big-sink} with
  $\apxO{(n - \outdegree{r}) m^{3/2 - \delta} / k}$. Choosing $k$ to
  balance this running time with the
  $\apxO{(n - \outdegree{r}) k^2 \log{n} / \eps^4}$ running time of
  \reflemma{apx-rvc-small-sink} gives the claimed running time.

  The running time for global vertex connectivity follows from rooted
  connectivity in the same way as \refcorollary{apx-vc} above.
\end{proof}
Balancing \reflemma{sparse-apx-vc} with \reftheorem{apx-rvc} and
\refcorollary{apx-vc} gives $\apxO{(mn)^{1-c} / \eps^2}$ running times
for approximating rooted and global vertex cuts, where $c> 0$ is a
constant, as claimed in \refcorollary{flow-barrier-apx-vc}.



\paragraph{Acknowledgements.} We thank Chandra Chekuri for helpful
discussion and detailed feedback.


\printbibliography

@article{benczur-karger,
  author =       {Andr{\'{a}}s A. Bencz{\'{u}}r and David R. Karger},
  title =        {Randomized Approximation Schemes for Cuts and Flows
                  in Capacitated Graphs},
  journal =      {{SIAM} J. Comput.},
  volume =       44,
  number =       2,
  pages =        {290--319},
  year =         2015,
}

@inproceedings{brand+20,
  author =       {Jan van den Brand and Yin Tat Lee and Danupon
                  Nanongkai and Richard Peng and Thatchaphol Saranurak
                  and Aaron Sidford and Zhao Song and Di Wang},
  title =        {Bipartite Matching in Nearly-linear Time on
                  Moderately Dense Graphs},
  booktitle =    {61st {IEEE} Annual Symposium on Foundations of
                  Computer Science, {FOCS} 2020, Durham, NC, USA,
                  November 16-19, 2020},
  pages =        {919--930},
  publisher =    {{IEEE}},
  year =         2020,
}

@article{brand+21,
  author =       {Jan van den Brand and Yin Tat Lee and Yang P. Liu
                  and Thatchaphol Saranurak and Aaron Sidford and Zhao
                  Song and Di Wang},
  title =        {Minimum Cost Flows, MDPs, and $\ell_1$-Regression in
                  Nearly Linear Time for Dense Instances},
  journal =      {CoRR},
  volume =       {abs/2101.05719},
  year =         2021,
  url =          {https://arxiv.org/abs/2101.05719},
  archivePrefix ={arXiv},
  eprint =       {2101.05719},
}

@article{cheriyan-reif,
  author =       {Joseph Cheriyan and John H. Reif},
  title =        {Directed $s$--$t$ Numberings, Rubber Bands, and
                  Testing Digraph $k$-Vertex Connectivity},
  journal =      {Comb.},
  volume =       14,
  number =       4,
  pages =        {435--451},
  year =         1994,
  addendum =     {Preliminary version in SODA, 1992}
}

@inproceedings{chuzhoy-khanna-19,
  author =       {Julia Chuzhoy and Sanjeev Khanna},
  editor =       {Moses Charikar and Edith Cohen},
  title =        {A new algorithm for decremental single-source
                  shortest paths with applications to
                  vertex-capacitated flow and cut problems},
  booktitle =    {Proceedings of the 51st Annual {ACM} {SIGACT}
                  Symposium on Theory of Computing, {STOC} 2019,
                  Phoenix, AZ, USA, June 23-26, 2019},
  pages =        {389--400},
  publisher =    {{ACM}},
  year =         2019
}

@Unpublished{cq-simple-connectivity,
  author =       {Chandra Chekuri and Kent Quanrud},
  title =        {Faster Algorithms for Rooted Connectivity
                  in~Directed~Graphs},
  note =         {Submitted},
  month =        {November},
  year =         2020
}

@InCollection{edmonds-70,
  author =       {Jack Edmonds},
  title =        {Submodular functions, matroids, and certain
                  polyhedra},
  booktitle =    {Combinatorial Structures and Their Applications
                  (Proceedings Calgary International Conference on
                  Combinatorial Structures and Their Applications,
                  Calgary, Alberta, 1969; , eds.)},
  publisher =    {Gordon and Breach},
  year =         1970,
  editor =       {R. Guy and H. Hanani and N. Sauer and
                  J. Sch\"{o}nheim},
  pages =        {69--87},
  address =      {New York},
}

@article{even-tarjan,
  author =       {Shimon Even and Robert Endre Tarjan},
  title =        {Network Flow and Testing Graph Connectivity},
  journal =      {{SIAM} J. Comput.},
  volume =       4,
  number =       4,
  pages =        {507--518},
  year =         1975
}

@article{ford-fulkerson,
  title =        {Maximal Flow Through a Network},
  volume =       8,
  DOI =          {10.4153/CJM-1956-045-5},
  journal =      {Canadian Journal of Mathematics},
  publisher =    {Cambridge University Press},
  author =       {Ford, L. R. and Fulkerson, D. R.},
  year =         1956,
  pages =        {399–404}
}

@inproceedings{forster+20,
  author =       {Sebastian Forster and Danupon Nanongkai and Liu Yang
                  and Thatchaphol Saranurak and Sorrachai
                  Yingchareonthawornchai},
  editor =       {Shuchi Chawla},
  title =        {Computing and Testing Small Connectivity in
                  Near-Linear Time and Queries via Fast Local Cut
                  Algorithms},
  booktitle =    {Proceedings of the 2020 {ACM-SIAM} Symposium on
                  Discrete Algorithms, {SODA} 2020, Salt Lake City,
                  UT, USA, January 5-8, 2020},
  pages =        {2046--2065},
  publisher =    {{SIAM}},
  year =         2020,
}

@Book{frank,
  author =       {Andras Frank},
  title =        {Connections in combinatorial optimization},
  publisher =    {Oxford University Press},
  year =         2011,
  series =       {Oxford Lecture Series in Mathematics and its
                  Applications}
}

@article{gabow-06,
  author =       {Harold N. Gabow},
  title =        {Using expander graphs to find vertex connectivity},
  journal =      {J. {ACM}},
  volume =       53,
  number =       5,
  pages =        {800--844},
  year =         2006,
  addendum =     {Preliminary version in FOCS, 2000}
}

@article{gabow-95,
  author =       {Harold N. Gabow},
  title =        {A Matroid Approach to Finding Edge Connectivity and
                  Packing Arborescences},
  journal =      {J. Comput. Syst. Sci.},
  volume =       50,
  number =       2,
  pages =        {259--273},
  year =         1995,
  addendum =     {Preliminary version in STOC 1991}
}

@article{galil-80,
  author =       {Zvi Galil},
  title =        {Finding the Vertex Connectivity of Graphs},
  journal =      {{SIAM} J. Comput.},
  volume =       9,
  number =       1,
  pages =        {197--199},
  year =         1980,
}

@article{glp-21,
  author =       {Yu Gao and Yang P. Liu and Richard Peng},
  title =        {Fully Dynamic Electrical Flows: Sparse Maxflow
                  Faster Than Goldberg-Rao},
  journal =      {CoRR},
  volume =       {abs/2101.07233},
  year =         2021,
  url =          {https://arxiv.org/abs/2101.07233},
  archivePrefix ={arXiv},
  eprint =       {2101.07233},
}

@article{goldberg-rao,
  author =       {Andrew V. Goldberg and Satish Rao},
  title =        {Beyond the Flow Decomposition Barrier},
  journal =      {J. {ACM}},
  volume =       45,
  number =       5,
  pages =        {783--797},
  year =         1998,
  addendum =     {Preliminary version in FOCS, 1997}
}

@article{hao-orlin,
  author =       {Jianxiu Hao and James B. Orlin},
  title =        {A Faster Algorithm for Finding the Minimum Cut in a
                  Directed Graph},
  journal =      {J. Algorithms},
  volume =       17,
  number =       3,
  pages =        {424--446},
  year =         1994,
  addendum =     {Preliminary version in SODA, 1992}
}

@article{hrg-00,
  author =       {Monika Rauch Henzinger and Satish Rao and Harold
                  N. Gabow},
  title =        {Computing Vertex Connectivity: New Bounds from Old
                  Techniques},
  journal =      {J. Algorithms},
  volume =       34,
  number =       2,
  pages =        {222--250},
  year =         2000,
  addendum =     {Preliminary version in FOCS, 1996}
}

@PhDThesis{karger,
  title =        "Random Sampling in Graph Optimization Problems",
  author =       "Karger, David R.",
  cat =          "Theory and Cuts and Flows",
  address =      "Stanford, CA 94305",
  year =         1994,
  date =         1994,
  department =   "Computer Science",
  school =       "Stanford University",
}

@inproceedings{kls-20,
  author =       {Tarun Kathuria and Yang P. Liu and Aaron Sidford},
  title =        {Unit Capacity Maxflow in Almost $O(m^{4/3})$ Time},
  booktitle =    {61st {IEEE} Annual Symposium on Foundations of
                  Computer Science, {FOCS} 2020, Durham, NC, USA,
                  November 16-19, 2020},
  pages =        {119--130},
  publisher =    {{IEEE}},
  year =         2020,
}

@inproceedings{lee-sidford,
  author =       {Yin Tat Lee and Aaron Sidford},
  title =        {Path Finding Methods for Linear Programming: Solving
                  Linear Programs in
                  $\tilde{O}\left({\sqrt{\rank}}\right)$ Iterations
                  and Faster Algorithms for Maximum Flow},
  booktitle =    {55th {IEEE} Annual Symposium on Foundations of
                  Computer Science, {FOCS} 2014, Philadelphia, PA,
                  USA, October 18-21, 2014},
  pages =        {424--433},
  year =         2014,
  note =         {Full versions available at
                  \url{http://arxiv.org/abs/1312.6677} and
                  \url{http://arxiv.org/abs/1312.6713}}
}

@misc{li+21,
  title =        {Vertex Connectivity in Poly-logarithmic Max-flows},
  author =       {Jason Li and Danupon Nanongkai and Debmalya
                  Panigrahi and Thatchaphol Saranurak and Sorrachai
                  Yingchareonthawornchai},
  year =         2021,
  eprint =       {2104.00104},
  archivePrefix ={arXiv},
  primaryClass = {cs.DS}
}

@Unpublished{li-21,
  author =       {Jason Li},
  title =        {Deterministic Mincut in Almost-Linear Time},
  note =         {To appear in STOC, 2021}
}

@inproceedings{liu-sidford-20,
  author =       {Yang P. Liu and Aaron Sidford},
  editor =       {Konstantin Makarychev and Yury Makarychev and Madhur
                  Tulsiani and Gautam Kamath and Julia Chuzhoy},
  title =        {Faster energy maximization for faster maximum flow},
  booktitle =    {Proccedings of the 52nd Annual {ACM} {SIGACT}
                  Symposium on Theory of Computing, {STOC} 2020,
                  Chicago, IL, USA, June 22-26, 2020},
  pages =        {803--814},
  publisher =    {{ACM}},
  year =         2020,
}

@inproceedings{madry-13,
  author =       {Aleksander Madry},
  title =        {Navigating Central Path with Electrical Flows: From
                  Flows to Matchings, and Back},
  booktitle =    {54th Annual {IEEE} Symposium on Foundations of
                  Computer Science, {FOCS} 2013, 26-29 October, 2013,
                  Berkeley, CA, {USA}},
  pages =        {253--262},
  publisher =    {{IEEE} Computer Society},
  year =         2013,
}

@inproceedings{madry-16,
  author =       {Aleksander Madry},
  editor =       {Irit Dinur},
  title =        {Computing Maximum Flow with Augmenting Electrical
                  Flows},
  booktitle =    {{IEEE} 57th Annual Symposium on Foundations of
                  Computer Science, {FOCS} 2016, 9-11 October 2016,
                  Hyatt Regency, New Brunswick, New Jersey, {USA}},
  pages =        {593--602},
  publisher =    {{IEEE} Computer Society},
  year =         2016
}

@article{mansour-schieber,
  author =       {Yishay Mansour and Baruch Schieber},
  title =        {Finding the Edge Connectivity of Directed Graphs},
  journal =      {J. Algorithms},
  volume =       10,
  number =       1,
  pages =        {76--85},
  year =         1989,
}

@inproceedings{matula-87,
  author =       {David W. Matula},
  title =        {Determining Edge Connectivity in $O(nm)$},
  booktitle =    {28th Annual Symposium on Foundations of Computer
                  Science, Los Angeles, California, USA, 27-29 October
                  1987},
  pages =        {249--251},
  publisher =    {{IEEE} Computer Society},
  year =         1987,
}

@inproceedings{nsy-19,
  author =       {Danupon Nanongkai and Thatchaphol Saranurak and
                  Sorrachai Yingchareonthawornchai},
  editor =       {Moses Charikar and Edith Cohen},
  title =        {Breaking quadratic time for small vertex
                  connectivity and an approximation scheme},
  booktitle =    {Proceedings of the 51st Annual {ACM} {SIGACT}
                  Symposium on Theory of Computing, {STOC} 2019,
                  Phoenix, AZ, USA, June 23-26, 2019},
  pages =        {241--252},
  publisher =    {{ACM}},
  year =         2019
}

@inproceedings{orlin-13,
  author =       {James B. Orlin},
  editor =       {Dan Boneh and Tim Roughgarden and Joan Feigenbaum},
  title =        {Max flows in $O(mn)$ time, or better},
  booktitle =    {Symposium on Theory of Computing Conference,
                  STOC'13, Palo Alto, CA, USA, June 1-4, 2013},
  pages =        {765--774},
  publisher =    {{ACM}},
  year =         2013
}

@article{podderyugin-73,
  author =       {V.~D. Podderyugin},
  title =        {An Algorithm for finding the edge connectivity of
                  graphs},
  journal =      {Vopr. Kibern.},
  year =         1973,
  volume =       2,
  pages =        136
}

@article{schnorr-79,
  author =       {Claus{-}Peter Schnorr},
  title =        {Bottlenecks and Edge Connectivity in Unsymmetrical
                  Networks},
  journal =      {{SIAM} J. Comput.},
  volume =       8,
  number =       2,
  pages =        {265--274},
  year =         1979,
}

@book{schrijver-book,
  author =       {A. Schrijver},
  publisher =    {Springer},
  series =       {Algorithms and Combinatorics},
  title =        {Combinatorial Optimization: Polyhedra and
                  Efficiency},
  volume =       24,
  year =         2003,
}

@article{spielman-srivastava,
  author =       {Daniel A. Spielman and Nikhil Srivastava},
  title =        {Graph Sparsification by Effective Resistances},
  journal =      {{SIAM} J. Comput.},
  volume =       40,
  number =       6,
  pages =        {1913--1926},
  year =         2011,
  addendum =     {Preliminary version in STOC, 2008}
}

@Article{timofeev-82,
  author =       {Eugeniy Aleksandrovich Timofeev },
  title =        {Algoritm postroeniya minimaksnogo k-svyaznogo orien-
                  tirovannogo podgrafa (Russian: An algorithm for
                  constructing minimax $k$-connected oriented graphs)},
  journal =      {Kibernetika },
  year =         1982,
  volume =       1982,
  number =       2,
  pages =        {109--110}
}

\end{document}